%% file: main.tex
\def\draft{1}
\documentclass{article}
\input{macros}

\title{A Fourier-Label Information-Loss\\ Barrier for Dihedral Coset Algorithms}
\author{Aparna Gupte\\MIT \and Seyoon Ragavan\\Google Quantum AI\\\& MIT \and Mark Zhandry\\Google Quantum AI \\\& Stanford University}
\date{}

\begin{document}

\maketitle

\begin{abstract}
We establish a no-go theorem for a broad class of quantum algorithms for the dihedral coset problem (DCP). We consider the Fourier-sampling and subset-sum-measurement template proposed by Regev (SIAM Journal on Computing, 2004), which is one of the main approaches to solving DCP. Suppose that, after measuring the lower $n-1$ bits of the subset sum, the algorithm discards any $\omega(\log n)$ bits from each of the Fourier labels. Then we prove that the algorithm \emph{cannot} succeed in solving DCP. This shows that any algorithm following this template must make extensive use of the Fourier labels, and thus serves as a useful guide for developing algorithms for DCP.

As a main application, we show that the recent algorithm by Simon (IACR ePrint:2026/1591, August 11 2026) does not solve DCP. We show that after the subset-sum measurement, this algorithm can be implemented (up to exponentially-small error) using only  the most-significant third of the Fourier labels, and is therefore subject to our general no-go theorem. 

To help with verifiability, we release Lean 4 code for our results, available at \href{https://github.com/sragavan99/lean-ePrint-2026-1591-refutation}{this GitHub repository.}
\end{abstract}

\section{Introduction}

A landmark work by Regev~\cite{DBLP:journals/siamcomp/Regev04} on the dihedral coset problem (DCP) showed that an efficient quantum polynomial-time (QPT) algorithm solving it would also yield a QPT algorithm for hard  lattice problems and undermine the security claims of proposed post-quantum lattice-based cryptography. 
In the same work, Regev proposed a Fourier-sampling and subset-sum-measurement template for solving DCP, and showed that it could be implemented given access to an oracle solving the subset-sum problem. The subset-sum oracle is believed to require super-polynomial time, however. Recently, Simon \cite{cryptoeprint:2026/1591} claimed a QPT algorithm for DCP via Regev's aforementioned template, claiming to remove the expensive subset-sum oracle.

In this work, we show a no-go theorem for a broad class of quantum algorithms for DCP following the template by Regev, including Simon's recent algorithm. In the following, we let $N = 2^n$ denote the DCP modulus and $d \in \Z_N$ the secret which the algorithm seeks to learn. The DCP algorithm, after Fourier sampling, is given samples of the form $(y, \frac{1}{\sqrt{2}}(\ket{0} + \omega_N^{dy} \ket{1}))$ for $y \gets \ZZ_N$.

\begin{theorem}[Informal, see~\Cref{thm:main} for precise statement]
    Consider any QPT algorithm that starts by following Regev's Fourier-sampling and subset-sum-measurement template, including measuring the lower $n-1$ bits of the subset sum. If at this point, the algorithm discards $\omega(\log n)$ bits of information about each Fourier label $y$, it cannot solve DCP.
\end{theorem}

An (un)fortunate consequence of our theorem, stated in~\Cref{cor:simon-does-not-work}, is that Simon's new algorithm cannot actually efficiently solve DCP, and therefore does not efficiently solve hard lattice problems nor undermine the security of lattice-based cryptosystems. The reason for this is that Simon's algorithm can be simulated (up to exponentially-small error) by discarding the $\approx 2n/3$ least-significant bits of each Fourier label $y$.

\paragraph{Outlook and caveats.}
Eliding some technical details, the algorithm by~\cite{DBLP:journals/siamcomp/Regev04} uses a subset-sum oracle to coherently erase the register of control bits, given the corresponding subset sum as input.
Implementing such a subset-sum oracle (inefficiently) requires the algorithm to use the $y$'s in their entirety.
The qualitative takeaway from our theorem is that this is an inherent requirement not just for Regev's algorithm but for any algorithm that follows Regev's template; even discarding as few as $\omega(\log n)$ bits of each $y_i$ in a predetermined way is enough to information-theoretically hide the least-significant bit of $d$. Given the generality of our no-go result, we are hopeful that it can extend beyond refuting the algorithm by~\cite{cryptoeprint:2026/1591} and actually help guide future efforts to develop quantum algorithms for DCP or related problems such as lattices.

We anticipate that there might be variants of the algorithm by~\cite{cryptoeprint:2026/1591} that circumvent the formal scope of Theorem~\ref{thm:main} (we describe some simple possibilities below in Remark~\ref{remark:superficial}), but in superficial ways that do not salvage the algorithm's correctness.
There is the possibility of a genuine and plausible repair path that circumvents Theorem~\ref{thm:main} in a meaningful way, but we are not aware of such an approach.

\paragraph{Organization.} We first recall the basic framework of Regev~\cite{DBLP:journals/siamcomp/Regev04} in \Cref{sec:algo-recap}, give intuition for why any algorithm in the framework that discards too many bits of the Fourier labels must fail in \Cref{sec:intuition}. We state the formal impossibility, our main result, in \Cref{main-result}, and then in \Cref{sec:simonspec} we apply the result to Simon's~\cite{cryptoeprint:2026/1591} algorithm to show that it cannot succeed. A brief background is given in \Cref{sec:prelim}. The proofs are in \Cref{sec:proof-main}.

\paragraph{Acknowledgements.}
The authors thank Noah Shutty, Daniel R. Simon, Vinod Vaikuntanathan, and Umesh Vazirani for helpful discussions.
AG's and SR's access to ChatGPT Pro was supported by the UK AISI alignment project.
SR's work on this note was supported in part by NSF CNS-2534400.

\paragraph{Statement on AI use.}
A preliminary version of Theorem~\ref{thm:main} that was more specific to the algorithm of~\cite{cryptoeprint:2026/1591} was obtained by the authors.
The authors then found a generalized no-go in conversation with GPT-5.6 Sol Pro and Ultra, and then used the same models to produce accompanying Lean 4 formalization.
The authors are fully accountable for the correctness of this note and the accompanying Lean formalization.

\section{Overview and Main Result}

\subsection{Recap of Regev's Template}
\label{sec:algo-recap}
\paragraph{Setup.}
Let $n$ be the security parameter, $N = 2^n$ the DCP modulus, $d \in \Z_N$ the DCP secret, and $m=kn^{c+1}$ the number of DCP Fourier samples\footnote{Simon's manuscript~\cite{cryptoeprint:2026/1591} instead uses $Q$ to denote the number of samples.} used by the algorithm as input for constants $k,c$.
One DCP Fourier sample consists of a uniformly random $y \gets \Z_N$ together with the state $\frac{1}{\sqrt{2}}(\ket{0} + \omega_N^{dy}\ket{1})$.
Here, $\omega_N = \exp(2\pi i/N)$.
The goal of the template put forth by~\cite{DBLP:journals/siamcomp/Regev04} is to extract the least-significant bit of $d$; by repeating this procedure in a standard way we could then recover all of $d$.

\paragraph{Steps 1 and 2 from the template of~\cite{DBLP:journals/siamcomp/Regev04}.}
In Step 1, the algorithm concatenates the $m$ Fourier samples to obtain:
\begin{center}\begin{tabular}{cc}
Classical info&Quantum state\\\\
$\yv$&
$\displaystyle|\psi^1_{\yv,d}\rangle\propto\sum_{\bv\in\{0,1\}^m}\omega_N^{d\langle \yv, \bv \rangle}|\bv\rangle$\end{tabular}\end{center}
The vector $\yv$ is sampled uniformly from $\Z_N^m$.

Now for Step 2, let $f_\yv(\bv)=\langle \yv, \bv \rangle\bmod N$. We will let $f^0_\yv(\bv)$ denote the highest-order bit of $f_\yv(\bv)$, and $f^1_\yv(\bv)$ denote the rest of the bits. The algorithm applies $f_\yv$ in superposition, and measures $f^1_\yv(\bv)$ to obtain the value $z$. Since $z$ is now classical, we can move this to the classical part of our information\footnote{Simon's manuscript~\cite{cryptoeprint:2026/1591} uses $z'$ instead of $z$ to denote the measurement result.}. By writing $f_\yv(\bv)=2^{n-1}f^0_\yv(\bv)+f^1_\yv(\bv)$, the phase term becomes
$$\omega_N^{df_\yv(\bv)}=\omega_N^{d(2^{n-1}f^0_{\yv}(\bv)+z)} = \omega_N^{dz} \cdot (-1)^{df_\yv^0(\bv)}.$$
The $\omega_N^{dz}$ part can be ignored since it is just an overall phase. Thus, letting $h = f_\yv^0(\bv)$, we can write the information obtained as: 

\begin{center}\begin{tabular}{cc}
Classical info&Quantum state\\\\
$\yv,z$&
$\displaystyle|\psi^2_{\yv,z,d}\rangle\propto \sum_{h \in \{0, 1\}} \sum_{\substack{\bv\in\{0,1\}^m \\ \langle \yv, \bv \rangle = z+hN/2}}(-1)^{dh}|\bv,h\rangle$\end{tabular}\end{center}
At this point, the reduction by~\cite{DBLP:journals/siamcomp/Regev04} assumes access to a subset-sum oracle to uncompute the $\bv$ register, leaving $\sum_h (-1)^{dh}|h\rangle$; measuring in the Hadamard basis then reveals the least-significant bit of $d$. (The algorithm of~\cite{cryptoeprint:2026/1591} aims to get a direct algorithm to recover the least-significant bit of $d$, meaning it cannot assume a subset-sum oracle. As such, it makes a necessary deviation which we will discuss in~\Cref{sec:simonspec}.)

Note that at this point our state depends only on the least-significant bit of $d$.

\subsection{Intuition for the No-Go}
\label{sec:intuition}

Before we formally state our no-go results, we provide some intuition for them.

To start, consider the state after Step 2. Suppose we measured the $\vecb$ register, collapsing it into a classical state. Since the secret $d$ lives in the phase of this state, collapsing the $\vecb$ register would lose all information about $d$ to the global phase. \textbf{Maintaining coherence across different values of $\bv$ is therefore crucial to recovering $d$.} In fact, we want to make this coherence computationally accessible---this is what Regev's algorithm~\cite{DBLP:journals/siamcomp/Regev04} does, assuming a subset-sum oracle. 

Our observation is that an algorithm must retain almost all of the information about the Fourier labels $\yv$ in order to maintain this coherence.

\paragraph{A Toy Example.} To build intuition, consider for a moment an algorithm that stops making use of $\yv, z$ after Step 2.
(We emphasize that this is not a requirement of our no-go, nor is it what \cite{cryptoeprint:2026/1591} is doing; we only start with this setting to build intuition.) If this were the case, then the algorithm's success would not be affected by instead discarding the $\vecy, z$ classical registers in their entirety. We can then compute the density matrix for the mixed state obtained by averaging over the discarded bits:
\begin{align*}
    \frac{1}{2^m N^m} \sum_{\substack{\vecy \in \ZZ_N^m \\ z \in \ZZ_{N/2} \\ h, h' \in \{0,1\}}} \sum_{\substack{\vecb, \vecb' \in \{0,1\}^m \\ f_\vecy(\vecb) = z + h N/2 \\ f_\vecy(\vecb') = z+h'N/2}} (-1)^{d (h - h')} \ketbra{ \vecb, h}{\vecb', h'}.
\end{align*}
In this hypothetical scenario, the algorithm must recover $d$ from this mixed state.

We claim that this state has negligible trace distance from the state obtained by simply measuring the $\vecb$ register. To see this, observe that outside the point $\bv=0$, the function $f_{\vecy}^1$ is a pairwise independent function if $\vecy$ is sampled randomly. The only exception to pairwise independence is the fact that $f_{\vecy}^1(0) = 0$. For each $(\vecb,\vecb')$ pair (excluding 0), over the choice of $\vecy$, the probability that $(\vecb,\vecb')$ remain coherent is exponentially small. Since $\vecy$ has been discarded and we average over the choice of $\vecy$, we find that the resulting density matrix has all off-diagonal terms exponentially suppressed. A direct calculation shows that the resulting density matrix is exponentially-close to the diagonal matrix--that is, the diagonal mixed state where the $\vecb$'s are totally decohered. We additionally need to handle the terms where $\vecb=0$, but since these have an exponentially-small mass, it does not significantly change the outcome. Thus, all information about $d$ is lost, except with exponentially-small probability.

A less precise version of the above intuition is the following: consider a purification of the algorithm that holds $\yv$ in superposition.
Then the $\yv$ register is very strongly entangled with the $\bv$ register.
If these registers were \emph{maximally} entangled, then discarding the $\yv$ register has the effect of completely decohering the $\bv$ register.
So if the registers were instead just strongly entangled, then discarding the $\yv$ register should still have a strong decoherence effect on the $\bv$ register.

\medskip

\paragraph{Applying to more general algorithms.} A more refined version of the argument above shows that even discarding a small amount of information about $\vecy$ after Step 2 is not enough to keep the $\vecb$ register usefully coherent. We formalize this using ``digest functions'' $H_i : \ZZ_N \to [K]$ for $i \in [m]$ that can be used to retain some information about $\vecy$ and discard the rest. Our main theorem shows that, if the digest functions have small enough outputs, it remains the case that $d$ is statistically hidden, even given the outputs of the digest function.

\subsection{Main Result}
\label{main-result}
We show in considerable generality that an algorithm that discards enough information about each coordinate of $\yv$ in this way cannot possibly work:
\begin{theorem}\label{thm:main}
    Fix any ``digest functions'' $H_1, \ldots, H_m: \Z_N \to [K]$. (These digest functions are fixed independently of the DCP secret and Fourier samples.)
    Then \textbf{any} algorithm that uses only the following information:
    \begin{center}\begin{tabular}{cc}
    Classical info&Quantum state\\\\
    $(H_1(y_1), \cdots, H_m(y_m)),z$&
    $\displaystyle|\psi^2_{\yv,z,d}\rangle\propto \sum_{h \in \{0, 1\}} \sum_{\substack{\bv\in\{0,1\}^m \\ \langle \yv, \bv \rangle = z+hN/2}}(-1)^{dh}|\bv,h\rangle$,\end{tabular}\end{center}
    where $\vecy, z, |\psi^2_{\vecy, z, d}\rangle$ are obtained by following Steps 1 and 2 described in \Cref{sec:algo-recap}, distinguishes between the cases where the least-significant bit of $d$ is 0 vs. 1 with advantage at most $m \sqrt{K/N}$.
    (See Definition~\ref{def:distadv} for a formal definition of distinguishing advantage.)
\end{theorem}

We prove the above in~\Cref{sec:proof-main}.

\subsection{Specializing to Simon's Algorithm}\label{sec:simonspec}

To see why~\Cref{thm:main} applies to the algorithm by~\cite{cryptoeprint:2026/1591}, we need to recall its next step after~\Cref{sec:algo-recap}. 
The reason is that this step makes use of the Fourier labels $\vecy$.

\paragraph{Step 3 of~\cite{cryptoeprint:2026/1591}.} This is where \cite{cryptoeprint:2026/1591} deviates from \cite{DBLP:journals/siamcomp/Regev04}.
Here, we divide $\bv$ and $\yv$ into $L = m/(c\log n)$ groups of $c\log n$ terms $\bv=(\bv^{(1)},\bv^{(2)},\cdots, \bv^{(L)})$, $\yv=(\yv^{(1)},\yv^{(2)},\cdots, \yv^{(L)})$ with $\bv^{(i)}\in\{0,1\}^{c\log n}$ and $\yv^{(i)}\in\Z_N^{c\log n}$. Let $g_{\yv^{(i)}}(\bv^{(i)})$ be the $\log n$ most significant bits of $\yv^{(i)}\cdot \bv^{(i)}\bmod N$.
In this step, we apply $g_{\yv^{(i)}}$ in superposition to the $i$th group, obtaining the information:

\begin{center}\begin{tabular}{cc}
Classical info&Quantum state\\\\
$\yv=(\yv^{(1)},\cdots,\yv^{(L)}),z$&
$\displaystyle|\psi^3_{\yv,z,d}\rangle\propto \sum_{h \in \{0, 1\}} \sum_{\substack{\bv = (\bv^{(1)}, \ldots, \bv^{(L)}) \in\{0,1\}^m \\ \langle \yv, \bv \rangle = z+hN/2}}(-1)^{dh}|\bv,h, g_{\yv^{(1)}}(\bv^{(1)}), \cdots, g_{\yv^{(L)}}(\bv^{(L)})\rangle$\end{tabular}\end{center}

\paragraph{Steps 4-7 of~\cite{cryptoeprint:2026/1591}.}
At this point, we will not need to recall what the rest of the algorithm does.
The key point is that from this point on, the algorithm of~\cite{cryptoeprint:2026/1591} works exclusively with the state $\ket{\psi_{\yv, z, d}^3}$ and does not make further use of the classical information $\yv, z$, so we can consider this information as being discarded.

Our observation is that in Step 3 of the algorithm by~\cite{cryptoeprint:2026/1591}, $\yv$ is only used to compute the functions $g_{\vecy^{(i)}}(\vecb^{(i)})$.
Since the $\vecb^{(i)}$ are small, changing the least-significant bits of $\vecy^{(i)}$ only slightly changes the inner-product, which means that the output of $g$ is unaffected (with high probability). As such, we can actually compute $g$, and hence the state in Step 3, only given, say, the top third most significant bits of $\yv$, discarding the least significant two-thirds as well as $z$.\footnote{As noted in~\Cref{remark:broad}, even if the algorithm additionally depended on $z$ it would not help.} The top-third of the bits of $\yv$ is a sufficiently small digest of $\yv$ for our theorem to apply. Due to the chance that low-order bits cause a large cascade of carries, our implementation using the top-third of the bits of $\yv$ does incur a tiny statistical error of $\poly(n)2^{-n/3}$, but this does not change the final result that all information about $d$ is hidden.
As such, {\bf Simon's algorithm and any possible algorithm that does not make use of the discarded bits of $\yv$ simply cannot recover any information about $d$.} We prove the below formally in~\Cref{sec:proof-main}.

\begin{corollary}\label{cor:simon-does-not-work}
    One run (i.e., a single run with $m$ samples that outputs one of $\{\mathsf{restart}, 0, 1\}$) of the algorithm by~\cite{cryptoeprint:2026/1591} distinguishes between the cases where the least-significant bit of $d$ is 0 vs. 1 with advantage at most $\poly(n)2^{-n/3}$.
    (See Definition~\ref{def:distadv} for a formal definition of distinguishing advantage.)
\end{corollary}

\paragraph{Remarks.}

\begin{remark}[Very superficial circumventions that do not salvage the algorithm]\label{remark:superficial}
    Theorem~\ref{thm:main} admits some very superficial circumventions by modifying the algorithm to retain a register containing $\yv$ (but not actually use this register in any way), or to postpone the measurement of $\yv$ until the end of the algorithm.
    In fact, the algorithm by Simon~\cite{cryptoeprint:2026/1591} appears to do this at the surface since it only measures $\yv$ at the end of Step 4.

    The specific tweaks described above do not protect the algorithm from our no-go results, because we can switch to an alternative algorithm that (a) is covered by Theorem~\ref{thm:main}; and (b) has exactly the same output distribution (in terms of the single-bit guess for $\mathsf{LSB}(d)$) as the original algorithm.
    Such an alternative algorithm can be obtained by discarding registers after the algorithm operates on them for the last time, and commuting the measurement of $\yv$ to the end of Step 1 (using the fact that we are only commuting this measurement past operations that control on the standard-basis contents of the $\yv$ register).
\end{remark}

\begin{remark}[Post-selection technicalities]
    Strictly speaking, the algorithm by~\cite{cryptoeprint:2026/1591} involves several runs, for both of the following reasons:
    \begin{itemize}
        \item To retry if one of the measurements in the later steps of the algorithm by~\cite{cryptoeprint:2026/1591} does not succeed. In this case, the run would output $\mathsf{restart}$; or
        \item To amplify the purported reliability of the algorithm's guess in $\{0, 1\}$ for the least-significant bit of the secret $d$.
    \end{itemize}
    The algorithm terminates when sufficiently-many runs output a bit instead of $\mathsf{restart}$.
    Without analyzing the remaining steps of the algorithm by~\cite{cryptoeprint:2026/1591}, we need to be careful to show that the additional runs do not circumvent Corollary~\ref{cor:simon-does-not-work}, which only applies to a single run.

    First, notice that if each run outputs $\mathsf{restart}$ with probability 1, then the expected run-time is infinite. Therefore, we assume the probability of $\mathsf{restart}$ is strictly less than 1. In this case, the algorithm will terminate with probability 1 and output a bit. 

    We next derive from Corollary~\ref{cor:simon-does-not-work} that there are only two possibilities: either the algorithm is simply incorrect, or the algorithm fails to run in even sub-exponential time. 
    This can be shown in a standard way as follows: fix two secrets $d_0,d_1$ with opposite least-significant bits, and
    consider a formally truncated version of the algorithm that is allowed at
    most $2^{\lfloor 2n/9\rfloor}$ runs, returning an arbitrary fixed bit if the algorithm has not decided its single-bit output by this point.
    These $2^{\lfloor 2n/9 \rfloor}$ runs should encompass both retries and amplification. We now have a win-win argument:
    \begin{itemize}
        \item By Corollary~\ref{cor:simon-does-not-work} and a standard
        hybrid argument, the formally truncated algorithm distinguishes the
        $\mathsf{LSB}(d)=0$ and $\mathsf{LSB}(d)=1$ cases with advantage at
        most
        \[
            2^{2n/9} \cdot\poly(n)2^{-n/3}
            \leq \poly(n)2^{-n/9}.
        \]

        \item Suppose that, for at least one $i\in\{0,1\}$, the eventual output
        distribution on secret $d_i$ has statistical distance at least
        $2^{-n/9}$ from its formally truncated output distribution. These
        distributions can differ only if the algorithm has not finished after $2^{\lfloor 2n/9 \rfloor}$ runs.
        Consequently, that event has probability at least $2^{-n/9}$, and
        the expected number of runs on secret $d_i$ is at least
        \[
            2^{-n/9} \cdot 2^{\lfloor 2n/9 \rfloor}
            \geq 2^{n/9-1}.
        \]
        Otherwise, for both secrets the eventual distribution is within
        $2^{-n/9}$ of its truncated distribution. The triangle inequality
        then shows that the two eventual output distributions have
        statistical distance at most
        \[
            \poly(n)2^{-n/9}+2\cdot2^{-n/9}
            =\poly(n)2^{-n/9}.
        \]
    \end{itemize}
    Thus, either the algorithm has expected runtime $\Omega(2^{n/9})$ for at least one of the two secrets, or its eventual distinguishing advantage between the two secrets is at most $\poly(n)2^{-n/9}$. In particular, it cannot simultaneously have subexponential expected runtime and inverse-subexponential correctness advantage.

\end{remark}

\begin{remark}\label{remark:broad}
    We remark that our no-go captures a broader class of algorithms than \cite{cryptoeprint:2026/1591} in two ways:
    \begin{enumerate}
        \item First, we rule out algorithms that retain \emph{arbitrary} small digests, not just the first few MSBs of blockwise subset sums $\vecg_{\vecy^{(i)}}$ that are used in \cite{cryptoeprint:2026/1591}.
        \item Second, we do not require that the algorithm discard $z = f^1_{\vecy}(\vecb)$ value; our no-go holds even if further operations that are classically controlled on this register are performed, as long as the $z$ register is collapsed. Simon's algorithm ignores the $z$ register after Step 2, so it effectively measures \emph{and discards} this register.
    \end{enumerate}
\end{remark}

\subsection{Related Work}

DCP arises naturally when designing algorothms for solving the dihedral hidden subgroup problem (DHSP). Indeed, almost all algorithms for DHSP work by generating DCP samples and then solving the DCP problem by beginning with Fourier sampling. Here, we give an overview of what we know about algorithms and no-gos for hidden subgroup problems, including but not limited to dihedral instances.

\paragraph{Subexponential-time algorithms for DCP.}
Kuperberg~\cite{kuperberg} gives a $2^{O(\sqrt{\log N})}$-time algorithm for DCP and hence DHSP, using subexponential space.
Regev~\cite{regevpolyspace} reduces the space requirement to $\poly(\log N)$ at the cost of a $2^{O(\sqrt{\log N\log\log N})}$ running time.
Kuperberg's later algorithm~\cite{kuperberg2013} achieves $2^{O(\sqrt{\log N})}$ time with $O(\log N)$ quantum space and subexponential classical space.
These algorithms process retained quantum registers using the classical Fourier labels to guide their operations, and use superpolynomially many samples.

Bonnetain and Naya-Plasencia~\cite{bonnetain} consider a variant where the secret lives in $\Z_{2^r}^n$, and show an algorithm running in time $n^{O(r)}$, which for $2^r\approx n$ is quasi-polynomial in the bit-length of the secret. This algorithm was later adapted by Bai et al.~\cite{baiedcp} to a variant known as the extrapolated dihedral coset problem in certain regimes. 
These algorithms do not yield any useful algorithms for lattices, however. This is roughly because these algorithms require a number of DCP samples proportional to their running time, and moreover the samples need to be exact. On the other hand, Regev's reduction~\cite{DBLP:journals/siamcomp/Regev04} can only generate samples that are noisy, and once the number of samples is super-polynomial, the noise accumulation is too much for the known DCP algorithms.

\paragraph{Sample-efficient algorithms for (D)HSP.}
The quantum algorithmic landscape for DCP, and HSP more generally, becomes significantly more positive if we consider the relaxed setting where the algorithm is only given polynomially many coset samples but otherwise need not be efficient. Ettinger, H{\o}yer, and Knill~\cite{ettingerhoyerknill,ettingerhoyer} show that the hidden subgroup problem over any finite group has quantum query complexity polynomial in the logarithm of the group order, via coset sampling followed by an inefficient measurement. 
Moreover, DHSP can be solved given polynomially many Fourier samples~\cite{baconchildsvandam,randomfourier}.

\paragraph{Prior work on information-loss barriers for HSP.} 
While we know that polynomially many coset samples suffice to solve HSP in any group~\cite{ettingerhoyerknill}, it is also known that polynomially many Fourier samples need not suffice. 
For certain symmetric-group instances relevant to graph isomorphism, Moore, Russell, and Schulman~\cite{moorerussellschulman} show that polynomially many samples are insufficient if each coset state is measured separately.
This includes strong Fourier sampling in any representation basis, and even arbitrary single-state measurements followed by unrestricted classical postprocessing.
Hallgren, Moore, R{\"o}tteler, Russell, and Sen~\cite{hallgrencoset} further establish lower bounds on how many coset states must be measured jointly to solve the graph-isomorphism instances using polynomially many samples.
Thus polynomially many coset states suffice for every finite group, but polynomially many individual strong Fourier samples suffice for the dihedral instances and fail for these symmetric-group instances. 

These no-go results are incomparable to ours, as we restrict attention to the dihedral setting but obtain our no-go by imposing additional restrictions that rule out the existing algorithms that work with polynomially many Fourier samples~\cite{baconchildsvandam,randomfourier}.

\section{Preliminaries}\label{sec:prelim}

We use $\norm{\cdot}_1$ to denote the trace norm (or the Schatten 1-norm) of a Hermitian operator.
It is multiplicative: $\norm{X \otimes Y}_1 = \norm{X}_1 \cdot \norm{Y}_1$, and additionally for (not necessarily unit) vectors $\ket{\phi}, \ket{\psi}$ we have:
\begin{equation}\label{eq:rank1}
    \norm{\ketbra{\phi}{\psi}}_1 = \norm{\ket{\phi}}_2 \cdot \norm{\ket{\psi}}_2
\end{equation}
\noindent
A (mixed) quantum state $\rho$ is a Hermitian PSD operator with trace 1.
The trace distance between two quantum states $\rho, \sigma$ is defined as $\frac{1}{2} \norm{\rho - \sigma}_1$.

\begin{definition}[Distinguishing advantage]\label{def:distadv}
    Given a quantum algorithm $\mathcal{A}$ with constant output space $\mathcal{U}$ and two quantum states $\rho, \sigma$, the distinguishing advantage of $\mathcal{A}$ between $\rho$ and $\sigma$ is defined to be $$\max_{\Phi: \mathcal{U} \to \{0, 1\}} \left|\Pr[\Phi(\mathcal{A}(\rho)) = 1] - \Pr[\Phi(\mathcal{A}(\sigma)) = 1]\right|.$$
    Equivalently, this is the total-variation distance between the distributions of $A(\rho)$ and $A(\sigma)$.
\end{definition}

\begin{lemma}[Distinguishing advantage is bounded by trace distance]\label{lem:tracedist}
    For any (even unbounded) algorithm $\mathcal{A}$, the distinguishing advantage of $\mathcal{A}$ between $\rho$ and $\sigma$ is at most $\frac{1}{2} \norm{\rho - \sigma}_1$.
\end{lemma}

\begin{lemma}[Norm of an average]
\label{lem:average-at-most-max-td}
Let $\mathcal{X}$ be a finite nonempty set, let $(p_x)_{x\in\mathcal{X}}$
be a probability distribution, and let $(A_x)_{x\in\mathcal{X}}$ be
finite-dimensional operators on the same space. Then
\begin{align*}
    \left\|\sum_{x\in\mathcal{X}}p_x A_x\right\|_1
    \leq \sum_{x\in\mathcal{X}}p_x\|A_x\|_1
    \leq \max_{x\in\mathcal{X}}\|A_x\|_1.
\end{align*}
In particular, for any two families of quantum states
$(\rho_x)_{x\in\mathcal{X}}$ and $(\sigma_x)_{x\in\mathcal{X}}$,
\begin{align*}
    \frac{1}{2}\left\|\sum_{x\in\mathcal{X}}p_x\rho_x
        -\sum_{x\in\mathcal{X}}p_x\sigma_x\right\|_1
    \leq \max_{x\in\mathcal{X}}
        \frac{1}{2}\|\rho_x-\sigma_x\|_1.
\end{align*}
\end{lemma}

\begin{proof}
The first inequality follows from the triangle inequality and absolute
homogeneity of the Schatten $1$-norm. The second follows from
$\sum_x p_x=1$. Applying the operator inequality to
$A_x=\rho_x-\sigma_x$ proves the final statement.
\end{proof}

\begin{definition}[Dephasing channel]
    Given a quantum state $\rho$ on two registers $R_1, R_2$, the dephasing channel applied to $R_1$ is defined as measuring the $R_1$ register (but not discarding the result).
    More formally, the dephasing operator is the channel that maps:
    \begin{equation*}
        \rho \mapsto \sum_x
\left(\ketbra{x}_{R_1} \otimes I_{R_2}\right)\rho
\left(\ketbra{x}_{R_1} \otimes I_{R_2}\right).
    \end{equation*}
\end{definition}

\section{Proofs}
\label{sec:proof-main}

\subsection{Proving \Cref{thm:main}}

Writing $\matH(\vecy) = (H_1(y_1), \ldots, H_m(y_m))$, for $d \in \{0,1\}$, we can write the state given as input to the algorithm in \Cref{thm:main} as follows:
\begin{align*}
    \rho_d &= 
    \frac{1}{2^mN^m} \sum_{\yv\in\Z_N^m} \sum_{z \in \Z_{N/2}} \sum_{\substack{\bv,\bv'\in\{0,1\}^m\\ f_\yv^1(\bv)=f_\yv^1(\bv') = z}}
(-1)^{d(f_\yv^0(\bv)-f_\yv^0(\bv'))} \ketbra{\matH(\vecy), z} \otimes \ket{\vecb, f_{\vecy}^0 (\vecb)} \bra{\vecb', f_{\vecy}^0 (\vecb')}.
\end{align*}

Define state $\rho_{d, {i}}$ as the state obtained by dephasing the $b_i$ register of $\rho_d$,
\begin{align*}
    \rho_{d, \textcolor{red}{i}} &= \frac{1}{2^mN^m} \sum_{\yv\in\Z_N^m} \sum_{z \in \Z_{N/2}} \sum_{\substack{\bv,\bv'\in\{0,1\}^m\\ f_\yv^1(\bv)=f_\yv^1(\bv') = z \\ \textcolor{red}{b_i = b_i'}}}
(-1)^{d(f_\yv^0(\bv)-f_\yv^0(\bv'))} \ketbra{\matH(\vecy), z} \otimes \ket{\vecb, f_{\vecy}^0 (\vecb)} \bra{\vecb', f_{\vecy}^0 (\vecb')}.
\end{align*}
The effect of dephasing the $i$th register is highlighted in red.
\begin{claim}
    \label{clm:dephase-i}
    For every $i \in [m]$ and $d \in \{0, 1\}$, it holds that
\begin{align*}
    \| \rho_d - \rho_{d, i} \|_1 \le \sqrt{\frac{K}{N}}.
\end{align*}
\end{claim}

\begin{proof}
Let $\vecy_{-i}$ denote the vector $\vecy$ with the $i$th entry set to $\bot$.
By \Cref{lem:average-at-most-max-td}, it suffices to bound $ \max_{\vecy_{-i}}  \| M(\vecy_{-i}) \|_1$, where $M(\yv_{-i})$ is defined as follows:
\begin{align*}
    M(\vecy_{-i}) &= \frac{1}{2^m N} \sum_{\substack{y_i \in \ZZ_N \\ z \in \ZZ_{N/2}}} \sum_{\substack{\vecb, \vecb' \in \{0,1\}^m\\ f_{\vecy}^1(\vecb) = f_{\vecy}^1(\vecb') = z \\ b_i \neq b_i'}} (-1)^{d (f_{\vecy}^0 (\vecb) - f_{\vecy}^0(\vecb')) } \ketbra{\matH(\vecy), z} \otimes \ketbra{\vecb, f^0_{\vecy}(\vecb)}{\vecb', f_{\vecy}^0 (\vecb')}.
\end{align*}
We will decompose $M (\vecy_{-i})$ into orthogonal blocks that are each just rank-$1$ operators, whose norms we can easily bound. These rank-$1$ operators are defined via the following vectors:
\begin{align*}
    \ket{\phi_{\red{0}}(\vecy_{-i}, z)} &= \sum_{\substack{\vecb \in \{0,1\}^m : \red{b_i = 0} \\ f_{\vecy}^1(\vecb) = z}} (-1)^{d f_{\vecy}^0(\vecb)} \ket{\vecb, f_{\vecy}^0(\vecb)}, \text{ and}\\
    \ket{\phi_{\red{1}} (\vecy_{-i}, z, a)} &= \sum_{y_i\in H_i^{-1}(a)}
\sum_{\substack{
    \bv\in\{0,1\}^m: \red{b_i=1}\\
    f_\yv^1(\bv)=z
}}
(-1)^{d f_\yv^0(\bv)}
\ket{\bv,f_\yv^0(\bv)}.
\end{align*}
The difference between the two vectors is highlighted in red.
Observe that since $\ket{\phi_0(\vecy_{-i}, z)}$ only sums over $\vecb$ where $b_i = 0$, the state does not depend on $y_i$.
To bound the $1$-norm of these rank-$1$ blocks, it will be useful to compute the $\ell_2$ norms of these vectors. 
\begin{align}
    \sum_{z \in \ZZ_{N/2}} \left\| \ket{\phi_0(\vecy_{-i}, z)} \right\|_2^2 &= \sum_{\substack{\vecb \in \{0,1\}^m : b_i = 0}} 1 = 2^{m-1}, \text{ and}\label{eq:phi0}\\
    \sum_{z \in \ZZ_{N/2}} \left\| \ket{\phi_1 (\vecy_{-i}, z, a)} \right\|_2^2 &= \sum_{y_i\in H_i^{-1}(a)}
    \sum_{\substack{
        \bv\in\{0,1\}^m: b_i=1}} 1 = |H_{i}^{-1}(a)| \cdot 2^{m-1},\label{eq:phi1}
\end{align}
where the second equality holds because the terms in $\ket{\phi_1 (\vecy_{-i}, z, a)}$ that arise from distinct values of $y_i$ have disjoint support: for contradiction, suppose that there exist $y \neq y'$ with an overlapping computational basis state. Let $\vecy$ be the vector formed by inserting $y$ in the $i$th location of $\vecy_{-i}$, and similarly let $\vecy'$ be the vector formed by inserting $y'$ in the $i$th location of $\vecy_{-i}$. Then, since there exists some overlapping computational basis state, there exists some $\vecb$ such that $b_i = 1$ and $f^0_{\vecy}(\vecb) = f^0_{\vecy'}(\vecb)$. Since $z = f^1_{\vecy}(\vecb) = f^1_{\vecy'}(\vecb)$ is also fixed, this means that $f_{\vecy}(\vecb) = f_{\vecy'}(\vecb)$, that is,
\begin{align*}
    y + \sum_{j \neq i} y_j b_j \equiv y' + \sum_{j \neq i} y_j b_j \pmod{N},
\end{align*}
which means that $y \equiv y' \pmod{N}$, a contradiction.

For convenience, let us use $\matH(\yv_{-i})$ to denote $\matH(\yv)$ after removing entry $i$; thus if $H_i(y_i) = a$ we can write $\matH(\yv) = (\matH(\yv_{-i}), a)$.
Then we can express $M(\vecy_{-i})$ in terms of the vectors $\ket{\phi_0(\vecy_{-i}, z)}$ and $\ket{\phi_1 (\vecy_{-i}, z, a)}$, to get
\begin{align*}
    \| M(\vecy_{-i}) \|_1 &= \frac{1}{2^m N} \left\| \sum_{a \in [K]} \sum_{z \in \ZZ_{N/2}} \ketbra{\matH(\vecy_{-i}), a, z}\otimes \Bigl( \ketbra{\phi_0(\vecy_{-i}, z)}{\phi_1 (\vecy_{-i}, z, a)} + \ketbra{\phi_1 (\vecy_{-i}, z, a)}{\phi_0(\vecy_{-i}, z)} \Bigr) \right\|_1\\
    &\le \frac{1}{2^m N} \sum_{a \in [K]} \sum_{z \in \ZZ_{N/2}} \norm{\ketbra{\matH(\vecy_{-i}), a, z}\otimes \Bigl( \ketbra{\phi_0(\vecy_{-i}, z)}{\phi_1 (\vecy_{-i}, z, a)} + \ketbra{\phi_1 (\vecy_{-i}, z, a)}{\phi_0(\vecy_{-i}, z)} \Bigr)}_1\\
    &= \frac{1}{2^m N} \sum_{a \in [K]} \sum_{z \in \ZZ_{N/2}} \norm{ \ketbra{\phi_0(\vecy_{-i}, z)}{\phi_1 (\vecy_{-i}, z, a)} + \ketbra{\phi_1 (\vecy_{-i}, z, a)}{\phi_0(\vecy_{-i}, z)}}_1 \\
    &\le \frac{1}{2^{m-1} N} \sum_{a \in [K]} \sum_{z \in \ZZ_{N/2}} \left\| \ket{\phi_0(\vecy_{-i}, z)} \right\|_2 \cdot \left\|\ket{\phi_1 (\vecy_{-i}, z, a)} \right\|_2 \qquad\text{ (plugging in~\eqref{eq:rank1})}\\
    &\le \frac{1}{2^{m-1} N} \sum_{a \in [K]} \left(\sqrt{\sum_{z \in \ZZ_{N/2}} \norm{\ket{\phi_0(\vecy_{-i}, z)}}_2^2}\sqrt{\sum_{z \in \ZZ_{N/2}} \norm{\ket{\phi_1 (\vecy_{-i}, z, a)}}_2^2}\right) \qquad\text{ (Cauchy-Schwarz)}\\
    &= \frac{1}{2^{m-1} N } \sum_{a \in [K]} 2^{m-1} \sqrt{|H^{-1}_i(a)|} \qquad\text{ (plugging in ~\eqref{eq:phi0} and~\eqref{eq:phi1})}\\
    &\le \frac{1}{N} \sqrt{K \cdot  \sum_{a \in [K]} |H^{-1}_i(a)|} \qquad\text{ (Cauchy-Schwarz)} \\
    &= \frac{1}{N} \sqrt{KN} \\
    &= \sqrt{\frac{K}{N}}.
\end{align*}
\end{proof}

\begin{proof}[Proof of Theorem~\ref{thm:main}]
We will finish with a hybrid argument that dephases one coordinate at a time. Let $\rho_{d, [i]}$ be the state obtained by applying the dephasing channel to the first $i$ qubits of $\vecb$. Therefore $\rho_{d, \emptyset} = \rho_d$. Applying the dephasing channel to the first $i-1$ qubit registers transforms $\rho_{d}$ to $\rho_{d, [i-1]}$ and $\rho_{d, i}$ to $\rho_{d, [i]}$. By the data processing inequality, the trace distance cannot increase by applying the dephasing channel, so
\begin{align*}
    \| \rho_{d, [i-1]} - \rho_{d, [i]}\|_1 \le \| \rho_{d} - \rho_{d, i}\|_1 \le \sqrt{K/N},
\end{align*}
where the last inequality holds because of \Cref{clm:dephase-i}.
Consequently, we have:
\begin{equation}
    \norm{\rho_d - \rho_{d, [m]}}_1 \leq \sum_{i = 1}^{m} \norm{\rho_{d, [i-1]} - \rho_{d, [i]}}_1 \leq m\sqrt{K/N}.
\end{equation}
Next, we claim that $\rho_{0, [m]} = \rho_{1, [m]}$.
Indeed, for either $d \in \{0,1\}$, after all $m$ bits of $\rho_d$ have been dephased, only the $\vecb = \vecb'$ terms remain, and since $(-1)^{d \cdot (f^0_{\vecy}(\vecb) - f_{\vecy}^0(\vecb))} = 1$ for all $\vecy$, the dependence on $d$ will be eliminated. Finally,
\begin{align*}
    \|\rho_0 - \rho_1\|_1 &\le \| \rho_0 - \rho_{0, [m]}\|_1 + \| \rho_1 - \rho_{1, [m]}\|_1\\
    &\le 2 m \sqrt{\frac{K}{N}}.
\end{align*}
Now applying Lemma~\ref{lem:tracedist} proves the theorem.
\end{proof}

\subsection{Proving \Cref{cor:simon-does-not-work}}

\begin{definition}[Digest functions]
    Write $r=\log n$ for the length of each Step-3 summary, and set $\ell=\left\lfloor 2n/3 \right\rfloor.$
For every coordinate, write
\begin{align*}
y_j=2^\ell q_j+u_j,
\qquad
0\le q_j<2^{n-\ell},
\qquad
0\le u_j<2^\ell,
\end{align*}
and choose the digest functions $H_j(y_j) = q_j+1$.
In words, the digest consists of the $n-\ell$ most-significant bits of $y_j$.
Thus the digest range has size $K = 2^{n-\ell}$.
\end{definition}

\par\medskip\noindent\emph{Reconstructing the Step-3 summaries.}
We first show that these digests suffice to reconstruct Simon's Step-3
summaries up to negligible error.  For block $i$, the exact summary is
\begin{align*}
\left\lfloor
\frac{2^r}{N}
\left(\left\langle\yv^{(i)},\bv^{(i)}\right\rangle\bmod N\right)
\right\rfloor,
\end{align*}
whereas the summary computed using only the digests is
\begin{align*}
\left\lfloor
\frac{2^r}{N}
\left(\left\langle2^\ell\vecq^{(i)},\bv^{(i)}\right\rangle\bmod N\right)
\right\rfloor.
\end{align*}

\begin{definition}
    Let $\mathsf{BAD} \subseteq \Z_N^m$ be the set of all vectors $\yv$ such that there exists a block $i \in [L]$ and a string $\bv^{(i)} \in \{0, 1\}^{c\log n}$ such that:
    \begin{equation}\label{eq:recongood}
        \left\lfloor
\frac{2^r}{N}
\left(\left\langle\yv^{(i)},\bv^{(i)}\right\rangle\bmod N\right)
\right\rfloor \neq \left\lfloor
\frac{2^r}{N}
\left(\left\langle2^\ell\vecq^{(i)},\bv^{(i)}\right\rangle\bmod N\right)
\right\rfloor.
    \end{equation}
\end{definition}

\begin{claim}\label{clm:ybad}
    We have:
    $$\Pr_{\yv \gets \Z_N^m} [\yv \in \mathsf{BAD}] \leq mn^{c+1} \cdot 2^{\ell-n}.$$
\end{claim}
\begin{proof}
    Let us begin by fixing an $i \in [L]$ and string $\bv^{(i)} \in \{0, 1\}^{c\log n}$.
    For this fixed $i, \bv^{(i)}$, we will bound the probability that~\eqref{eq:recongood} is violated.
    If $\bv^{(i)} = 0$, the LHS and RHS are clearly both 0, so assume henceforth that this is not the case.
    At the end, we will take a union bound.
    Note that:
    \begin{align*}
    \yv^{(i)}=2^\ell\vecq^{(i)}+\vecu^{(i)},
    \end{align*}
    and hence
    \begin{align*}
    \left\langle\yv^{(i)},\bv^{(i)}\right\rangle
    ={}&\left\langle2^\ell\vecq^{(i)},\bv^{(i)}\right\rangle
    +\left\langle\vecu^{(i)},\bv^{(i)}\right\rangle\pmod N,\\
    0\le\left\langle\vecu^{(i)},\bv^{(i)}\right\rangle
    &\le\frac mL(2^\ell-1).
    \end{align*}
    The exact and reconstructed summaries can therefore differ only if adding
    $\langle\vecu^{(i)},\bv^{(i)}\rangle$ crosses one of the $2^r$ boundaries
    between consecutive summary values.
    So the number of bad values of $\langle \yv^{(i)}, \bv^{(i)} \rangle \bmod{N}$ is at most:
    $$2^r \cdot \frac{m}{L}(2^\ell - 1) < \frac{m}{L} \cdot 2^{r+\ell}.$$
    Since $\bv^{(i)}$ is nonzero, $\langle \yv^{(i)}, \bv^{(i)} \rangle \bmod{N}$ is uniformly distributed $\bmod N$ so the probability that~\eqref{eq:recongood} is violated is $\leq \frac{m}{L} \cdot 2^{r+\ell-n}$.
    Finally, taking a union bound over the $L$ values of $i$ and the $2^{c\log n}$ values of $\bv^{(i)}$ for each $i$ implies the conclusion.
\end{proof}

\begin{claim}\label{clm:simonrecon}
    Let $\rho_d^{(3)}$ be Simon's exact Step-3 state after discarding the Fourier labels, and let $\widehat{\rho_d^{(3, \ell)}}$ be obtained by computing the summaries from the digests $q_1, \ldots, q_m$ and then discarding those digests.
    Then we have:
    $$\frac{1}{2} \norm{\rho_d^{(3)} - \widehat{\rho_d^{(3, \ell)}}}_1 \leq mn^{c+1} \cdot 2^{\ell-n}.$$
\end{claim}
\begin{proof}
    If $\Pr[\yv \in \mathsf{BAD}]$ is 0 then these states are identical, so assume from now that this is not the case.
    A distinguisher between $\rho_d^{(3)}$ and $\widehat{\rho_d^{(3, \ell)}}$ can only succeed when $\yv \in \mathsf{BAD}$.
    A little more formally, we can apply the first inequality in Lemma~\ref{lem:average-at-most-max-td} after splitting each state into a probabilistic mixture of the conditional mixed state when $\yv \in \mathsf{BAD}$ and the conditional mixed state when $\yv \notin \mathsf{BAD}$.
    In the latter case, the conditional mixed states are identical.
    In the former case, the conditional mixed states have trace distance at most 1.
    So the overall trace distance is at most the probability that $\yv \in \mathsf{BAD}$, which by Claim~\ref{clm:ybad} is at most the stated bound.
\end{proof}

\begin{proof}[Proof of \Cref{cor:simon-does-not-work}]
Consider the algorithm that starts with
\begin{align*}
(H_1(y_1),\ldots,H_m(y_m)),z
\qquad\text{and}\qquad
\ket{\psi^2_{\yv,z,d}},
\end{align*}
computes the reconstructed Step-3 summaries, and then performs Steps 4--7 of
Simon's algorithm.  This algorithm uses exactly the information allowed in
\Cref{thm:main}.  Since $K=2^{n-\ell}$, that theorem bounds its advantage
by
\begin{align*}
\poly(n)\sqrt{\frac KN}
&=\poly(n)
2^{-\ell/2}\\
&\le\poly(n)2^{-n/3+1}\\
&=\poly(n)2^{-n/3}.
\end{align*}
\noindent
Finally, a quantum channel cannot increase Schatten $1$-norm distance.  Thus
replacing the reconstructed Step-3 state by Simon's exact Step-3 state
changes the final guessing advantage by at most
\begin{align*}
O\left(\frac12\max_{d\in\{0,1\}}
\left\|\rho_d^{(3)}-
\widehat\rho_d^{(3,\ell)}\right\|_1\right)
\le\poly(n)2^{\ell-n}\leq \poly(n)2^{-n/3},
\end{align*}
by Claim~\ref{clm:simonrecon}.
Adding the reconstruction error to the bound from \Cref{thm:main} proves
that Simon's algorithm has advantage at most
$\poly(n)2^{-n/3}$.
\end{proof}

\bibliographystyle{alpha}
\bibliography{refs}

\end{document}

%% file: macros.tex
\usepackage{fullpage}
\usepackage{graphicx}
\usepackage{amsmath,amsfonts,amssymb}
\usepackage{amsthm}
\usepackage{lmodern}
\usepackage{txfonts}
\usepackage[pdfstartview=FitH,colorlinks,linkcolor=blue,citecolor=blue]{hyperref}

\theoremstyle{plain} 

\newcommand{\Z}{{\mathbb{Z}}}

\newcommand{\bv}{{\mathbf{b}}}

\newcommand{\yv}{{\mathbf{y}}}

\newcommand{\ket}[1]{{\left|#1\right\rangle}}

\newcommand{\poly}{{\sf poly}}

\usepackage{fullpage}
\usepackage[normalem]{ulem}
\usepackage{libertine}

\usepackage{amsmath}
\usepackage{amssymb}
\usepackage{amsthm}
\usepackage{mathtools}
\usepackage{thmtools}
\usepackage{bbm}
\usepackage[T1]{fontenc}

\usepackage{graphicx}
\usepackage{float}
\usepackage[dvipsnames]{xcolor}
\usepackage{tikz}
\usepackage{subcaption}
\usepackage{multirow}
\usepackage{enumitem}

\usepackage{algpascal}
\usepackage[vlined, ruled]{algorithm2e}

\usetikzlibrary{quantikz2}
\usepackage{physics}

\usepackage{cleveref}

\newcommand{\ZZ}{\mathbb{Z}}

\newcommand{\matH}{\mathbf{H}}

\newcommand{\vecb}{\mathbf{b}}

\newcommand{\vecg}{\mathbf{g}}

\newcommand{\vecq}{\mathbf{q}}

\newcommand{\vecu}{\mathbf{u}}

\newcommand{\vecy}{\mathbf{y}}

\renewcommand{\epsilon}{\varepsilon}
\renewcommand{\phi}{\varphi}

\ifnum\draft=1
\newcommand{\anote}[1]{\textcolor{Purple}{\textit{\textbf{Aparna}: #1}}}

\else
\newcommand{\anote}[1]{}
\fi

\newtheorem{theorem}{Theorem}[section]
\newtheorem{lemma}[theorem]{Lemma}

\newtheorem{definition}[theorem]{Definition}
\newtheorem{claim}[theorem]{Claim}
\newtheorem{corollary}[theorem]{Corollary}
\newtheorem{remark}[theorem]{Remark}

\allowdisplaybreaks

\newcommand{\red}[1]{\textcolor{red}{#1}}